\documentclass[12pt]{article}

\usepackage{a4wide, amsmath,amsthm,amsfonts,amscd,amssymb,eucal,bbm,mathrsfs, tikz}
\usetikzlibrary{matrix,arrows}

\def\RR{{\mathbb R}}
\def\CC{{\mathbb C}}

\def\ZZ{{\mathbb Z}}

\def\A{{\mathcal A}}
\def\B{{\mathcal B}}

\def\F{{\mathcal F}}
\def\H{{\mathcal H}}

\def\K{{\mathcal K}}
\def\M{{\mathcal M}}
\def\N{{\mathcal N}}

\def\P{{\mathcal P}}
\def\R{{\mathcal R}}

\def\a{\alpha}

\def\f{\varphi}

\def\G{\Gamma}

\def\k{\kappa}

\def\L{{\mathrm L}}

\def\p{\psi}

\def\R{{\mathrm R}}
\def\Rc{{\mathcal R}}

\def\t{\tau}

\def\Ad{{\hbox{\rm Ad\,}}}

\def\sp{{\rm sp}\,}

\def\1{{\mathbbm 1}}

\def\netu1{{\A^{(0)}}}

\def\diff{{\rm Diff}}

\def\diffs1{\diff(S^1)}
\def\mob{{\rm M\ddot{o}b}}
\def\mob2{{\rm M\ddot{o}b}^{(2)}}

\def\supp{{\rm supp}}

\def\u1{{\rm U}(1)}
\def\psl2r{{\rm PSL}(2,\RR)}
\def\sl2r{{\rm SL}(2,\RR)}
\def\su11{{\rm SU}(1,1)}
\def\2dmob{{\overline{\psl2r}\times\overline{\psl2r}}}
\def\<{\langle}
\def\>{\rangle}

\def\poincare{{\P^\uparrow_+}}

\newcommand{\Ar}{\mathcal{A}_{\mathrm{r}}} 
\newcommand{\Mr}{\mathcal{M}_{\mathrm{r}}} 
\newcommand{\Hr}{\mathcal{H}_{\mathrm{r}}} 
\newcommand{\Ur}{U_{\mathrm{r}}}
\newcommand{\Trr}{T_{\mathrm{r}}} 
\newcommand{\Omr}{\Omega_{\mathrm{r}}} 
\newcommand{\Mtr}{\widetilde{\mathcal{M}}_{\mathrm{r},\f}} 
\newcommand{\Rt}{\widetilde{R}_\f} 
\newcommand{\Rtb}{\widetilde{R}_{\overline\f}} 
\newcommand{\Htr}{\widetilde{\mathcal{H}}_{\mathrm{r}}} 

\newcommand{\Ttrr}{\widetilde{T}_{\mathrm{r}}} 
\newcommand{\Omtr}{\widetilde{\Omega}_{\mathrm{r}}} 
\newcommand{\Str}{\widetilde{S}_{\mathrm{r},\f}} 
\newcommand{\Hc}{\mathcal{H}_{\mathrm{c}}} 
\newcommand{\Tc}{T_{\mathrm{c}}}
\newcommand{\Qc}{Q_{\mathrm{c}}}
\newcommand{\Mc}{\mathcal{M}_{\mathrm{c}}} 
\newcommand{\Nc}{\mathcal{N}_{\mathrm{c}}} 
\newcommand{\Rcc}{\mathcal{R}_{\mathrm{c}}} 
\newcommand{\Vc}{V_{\mathrm{c}}} 
\newcommand{\Jc}{J_{\mathrm{c}}} 
\newcommand{\Dc}{\Delta_{\mathrm{c}}} 
\newcommand{\Omc}{\Omega_{\mathrm{c}}} 
\newcommand{\ac}{\a_{\mathrm{c}}}
\newcommand{\acs}{\a_{\mathrm{c},\k}}
\newcommand{\Htc}{\widetilde{\mathcal{H}}_{\mathrm{c}}} 
\newcommand{\Ttc}{\widetilde{T}_{\mathrm{c}}}
\newcommand{\Mtc}{\widetilde{\mathcal{M}}_{\mathrm{c},k}} 
\newcommand{\Rtcc}{\widetilde{\mathcal{R}}_{\mathrm{c},k}} 
\newcommand{\Rtccs}{\widetilde{\mathcal{R}}_{\mathrm{c},\k}} 
 
\newcommand{\Stcs}{\widetilde{S}_{\mathrm{c},\k}} 
\newcommand{\Omtc}{\widetilde{\Omega}_{\mathrm{c}}}

\newcommand{\Vtc}{\widetilde{V}_{\mathrm{c},k}}
\newcommand{\Vtcs}{\widetilde{V}_{\mathrm{c},\k}}
\newcommand{\Mtcs}{\widetilde{\mathcal{M}}_{\mathrm{c},\k}}

\newtheorem{theorem}{Theorem}[section]

\newtheorem{proposition}[theorem]{Proposition}
\newtheorem{lemma}[theorem]{Lemma}
\theoremstyle{remark}

\title{Construction of two-dimensional quantum field models through Longo-Witten endomorphisms}
\date{}
\author{
{\bf Yoh Tanimoto} \footnote{Supported by
Deutscher Akademischer Austauschdienst, Hausdorff Institut f\"ur Mathematik,
Alexander von Humboldt Stiftung and Grant-in-Aid for JSPS fellows 25-205, and in part by
Courant Research Centre ``Higher Order Structures in Mathematics'', G\"ottingen.} \\
e-mail: {\tt yoh.tanimoto@theorie.physik.uni-goettingen.de}\\
Graduate school of mathematical sciences, University of Tokyo / \\
Institut f\"ur Theoretische Physik, Universit\"at G\"ottingen \\
Friedrich-Hund-Platz 1, 37077 G\"ottingen, Germany.\\
}
\begin{document}
\maketitle
\begin{abstract}
We present a procedure to construct families of local, massive and interacting Haag-Kastler nets
on the two-dimensional spacetime through an operator-algebraic method.
An existence proof of local observables is given without relying on modular nuclearity.

By a similar technique, another family of wedge-local nets is constructed using certain endomorphisms of
conformal nets recently studied by Longo and Witten.
\end{abstract}

\section{Introduction}\label{introduction}

In a series of papers \cite{DT11, Tanimoto12-2, BT12} we have investigated
operator-algebraic methods based on conformal field theory to construct quantum
field models on two-dimensional spacetime with a weak localization property.
Although we succeeded to obtain various
examples and general structural results, there was missing one important property:
strict localization. Namely, we have constructed certain operator-algebraic objects
which are considered to represent observables localized in the wedge-shaped regions
in two-dimensional spacetime. But we could not prove the existence of observables in
compactly localized regions.
In the present paper we construct further new families of quantum
field models and prove their strict locality. In other words, we construct interacting
two-dimensional Haag-Kastler nets, using techniques from conformal nets.

In recent years the operator-algebraic approach (algebraic QFT) to quantum field theory has seen
many developments. A fundamental idea is that the whole quantum field model can
be recovered from the set of observables localized in an unbounded wedge-shaped
region and the spacetime translations \cite{Borchers92}. Conversely, a strategy to
construct models is first to obtain models localized in wedges then to prove the
existence of fully localized observables. We say that they are wedge-local and 
strictly local, respectively.
The most successful application of this strategy is the construction of scalar
factorizing S-matrix models in two dimensions \cite{Lechner08}.

Construction of interacting quantum field theory in four dimensions remains one
of the most important open problems in mathematical physics. Further steps toward
higher dimensions within the above operator-algebraic approach have been obtained
\cite{BLS11, Lechner11}, however, the second step to find localized observables
has so far turned out to be unsuccessful \cite{BLS11}.
Furthermore, we found that conformal covariance implies triviality of scattering in two
dimensions \cite{Tanimoto12-1} yet it is possible to construct weakly localized,
conformally covariant models \cite{Tanimoto12-2}. Thus construction of wedge-local
models should be considered as an important but intermediate step.

We have constructed weakly localized massless models using chiral components of
two-dimensional conformal field theory \cite{DT11, Tanimoto12-2, BT12},
however, strict locality is disproved in some of them
and open in the others. The only known operator-algebraic technique to find local observables
is wedge-splitting \cite{BL04} (see Section \ref{general}). This property is apparently difficult to hold
in massless models, especially known to be invalid in conformal case.
Actually, one encounters a similar problem in the form factor bootstrap
program to integrable models \cite{Smirnov92}: one starts with a given S-matrix and
matrix coefficients of local operators are given as solutions of the so-called form factor equations.
One can easily observe that the convergence of these matrix coefficients is much worse
in massless case. Hence, in order to avoid such technical difficulty, one should
try to construct massive models (besides, the convergence of form factors is a
difficult problem even in massive models \cite{BK04}).

Although our previous constructions are fundamentally based on conformal field theory,
one can still find a connection with massive models.
We will investigate this issue more systematically in a separate paper \cite{BT13}.
Here we use again Longo-Witten endomorphisms \cite{LW11} to construct weakly
localized massive models. Then we show that some of them are actually strictly local by
examining wedge-splitting. Note that in \cite{Lechner08} wedge-splitting has been proved through
modular nuclearity \cite{BL04}. We provide both a direct proof of wedge-splitting and
a simple proof of modular nuclearity of certain models.
Once strict locality is demonstrated, one can apply
the standard scattering theory and check that S-matrix is nontrivial.
As we will see, this construction can be considered as a
generalization of the Federbush model. The Federbush model is an integrable model with S-matrix
independent from rapidity \cite{Ruijsenaars82}. Local observables in this model have been found in
\cite{Ruijsenaars83} for certain range of coupling constant.
We construct Haag-Kastler nets which have the same S-matrix as those Federbush models
with the coupling constant of arbitrary size.
The model describes multiple species of particles, hence the strict locality has not been treated in
\cite{Lechner08}. We construct another family of wedge-local nets (Borchers triples)
using Longo-Witten endomorphisms found in \cite{LW11}.
Although strict locality is open for this latter family, two constructions look quite similar.

The present construction uses as a starting point the free field or
any quantum field which gives a wedge-split net.
The new models are constructed on the tensor
product of the Hilbert space of the input field (net).
One could call our construction a ``deformation'' of the trivial combination of
the input fields. The term ``deformation'' had a specific sense in \cite{BLS11} and \cite{Lechner11},
and in massless cases it has been
revealed that the connection with free (trivial scattering) models is generic \cite{Tanimoto12-2}.
In addition, one notes that such  a ``deformation'' is always related to a certain
structure of the input field. The deformation in \cite{BLS11} exploits the translation covariance,
while the chiral decomposition plays a crucial role in \cite{Tanimoto12-2, BT12}.
If one desires further different deformations, it is a natural idea to
introduce a further structure to the ``undeformed'' model. In this paper this idea will be realized using
tensor product and inner symmetry. Such structure is easily found in examples. Indeed,
our procedure in Section \ref{inner} can be applied to the tensor product of Lechner's models \cite{Lechner08}
and immediately gives new strict local nets.

This paper is organized as follows: after introductory Sections \ref{introduction}, \ref{preliminaries},
we first exhibit a way to construct Borchers triples. Section \ref{inner} is of general
nature: given a Borchers triple with certain symmetry, we show a way to
construct new Borchers triples. 
Then we turn to strict locality in Section \ref{general}.
We collect arguments to prove strict locality of the models from Section \ref{inner}.
This is our main result.
Then we apply them to the concrete cases including the complex massive free field to
obtain interacting Haag-Kastler nets (Section \ref{scattering}).
Section \ref{u1} demonstrates how to construct
massive Borchers triples out of Longo-Witten endomorphisms of the $\u1$-current net.
The two-particle S-matrix of these models will be also calculated. We do not investigate
strict locality of these models.

\section{Preliminaries}\label{preliminaries}

\subsection{Two dimensional Haag-Kastler net and Borchers triples}
In algebraic QFT, models of quantum field theory are realized as nets of von Neumann algebras
\cite{Haag96}. A {\bf Haag-Kastler net (of von Neumann algebras)} $\A$ on $\RR^2$ is a map
$O\mapsto \A(O)$ from the set of open regions $\{O\}$ in two dimensional Minkowski space
$\RR^2$ to the set of von Neumann algebras on a fixed Hilbert space $\H$ such that
\begin{enumerate}
\item[(1)] {\bf (Isotony)} If $O_1 \subset O_2$, then $\A(O_1) \subset \A(O_2)$.
\item[(2)] {\bf (Locality)} If $O_1$ and $O_2$ are spacelike separated, then $\A(O_1)$ and $\A(O_2)$
commute.
\item[(3)] {\bf (Poincar\'e covariance)} There is a (strongly continuous) unitary representation $U$ of
the (proper orthochronous) Poincar\'e group $\poincare$ such that $U(g)\A(O)U(g)^* = \A(gO)$
for $g \in \poincare$.
\item[(4)] {\bf (Positivity of energy)} The joint spectrum of the generators of the translation
subgroup ($\cong \RR^2$) in the representation $U$ is contained in the closed forward lightcone
$V_+ := \{(p_0,p_1) \in \RR^2: p_0 \ge |p_1|\}$.
\item[(5)] {(\bf Vacuum vector)} There is a (up to scalar) unique vector $\Omega$ which is invariant
under $U(g)$ and cyclic for $\A(O)$.
\end{enumerate}
Note that the condition on the vacuum $\Omega$ contains the Reeh-Schlieder property,
which can be derived if one assumes the additivity of the net. We take this as an axiom for net
for simplicity (see the discussion in \cite[Section 2]{Weiner11}). From these assumptions,
the following is automatic \cite{Baumgaertel}.
\begin{enumerate}
\item[(6)] {\bf (Irreducibility)} The von Neumann algebra $\bigvee_{O \subset \RR^2} \A(O)$ is
equal to $B(\H)$.
\end{enumerate}
Precisely, the triple $(\A, U, \Omega)$ should be called a Haag-Kastler net, however we say $\A$ is a
net if no confusion arises.
If one has a net $\A$, under certain conditions one can define S-matrix (see Section
\ref{scattering}). The main objective in this paper is to construct examples of two-dimensional
Haag-Kastler nets with nontrivial S-matrix.

Yet, it appears very difficult to construct such nets from a scratch.
Fortunately, Borchers showed that it suffices to have a triple of a {\it single}
von Neumann algebra associated with the wedge-shaped region $W_\R := \{a\in\RR^2: a_1 > |a_0|\}$,
the spacetime translations and the vacuum.
More precisely, a {\bf Borchers triple} is a triple $(\M,T,\Omega)$ of a von Neumann algebra
$\M$, a unitary representation $U$ of $\RR^2$ and a vector $\Omega$ such that
\begin{enumerate}
 \item[(1)] If $a\in W_\R$, then $\Ad T(a)(\M) \subset \M$.
 \item[(2)] The joint spectrum of the generators of $T$ is contained in $V_+$.
 \item[(3)] $\Omega$ is cyclic and separating for $\M$.
\end{enumerate}
The correspondence between nets and Borchers triples is as follows (see \cite{Borchers92, Lechner08}): 
If $(\A, U, \Omega)$ is a (two-dimensional) Haag-Kastler net, then $(\A(W_\R), U|_{\RR^2}, \Omega)$
is a Borchers triple, where $U$ is restricted to the translation subgroup $\cong \RR^2$.

On the other hand, if $(\M, T, \Omega)$ is a Borchers triple, then
one can define a net as follows: in the two-dimensional Minkowski space any double cone $D$ is
represented as the intersection of two wedges $D = (W_\R + a) \cap (W_\L + b)$, where
$W_\L$ is the standard left wedge $W_\L := \{a\in\RR^2: -a_1 > |a_0|\}$. Then we define
for double cones $\A(D) := \Ad T(a)(\M) \cap \Ad T(b)(\M')$. For a general region $O$ we take
$\A(O) := \bigvee_{D \subset O} \A(D)$, where the union runs over all the double cones contained in
$O$. Then one can observe that $\A$ satisfies isotony and locality.
Borchers further proved that the representation $T$ of $\RR^2$ extends to a representation $U$ of
$\poincare$ through Tomita-Takesaki theory \cite{Borchers92} such that $\A$ is covariant
and $\Omega$ is invariant under $U$ (the representation of boosts is given
by the modular group). Positivity of energy is inherited from $T$ of the Borchers triple.
The only missing property is cyclicity of $\Omega$ for local algebras $\A(O)$.
Indeed, there are examples of Borchers triples which fail to satisfy this cyclicity \cite[Theorem 4.16]{Tanimoto12-2}.

Hence, a general strategy to construct Haag-Kastler nets is first to construct Borchers triples
and then to check cyclicity of the vacuum. This program has been first completed in \cite{Lechner08}.
Note that the latter condition is actually very hard to check directly. In \cite{Lechner08}, Lechner
proved instead the so-called modular nuclearity \cite{BL04},
which is a sufficient condition for the cyclicity of the vacuum. In this paper, we will provide
new examples of Borchers triples for which the cyclicity of the vacuum can be proved without relying
on modular nuclearity (see Section \ref{split}). We say that a Borchers triple is {\bf strictly local}
if $\Omega$ is cyclic for $\A(O)$ constructed as above. A strictly local Borchers triple
corresponds to a Haag-Kastler net.

If there is a unitary operator $V$ which commutes with $U(g)$ and if $\Ad V$ preserves
each local algebra $\A(O)$, then we call $\Ad V$ an {\bf inner symmetry} of the net $\A$.
By the uniqueness of $\Omega$, $\Ad V$ preserves the vacuum state $\<\Omega, \cdot\,\Omega\>$.
By the irreducibility, the automorphism $\Ad V$ is implemented uniquely by $V$ up to a scalar.
We require always that $V\Omega = \Omega$. By this requirement, the implementation is unique.
A general strategy to use them in order to construct new Borchers triples will be explained
in Section \ref{inner} and concrete examples of inner symmetry will be discussed in Section \ref{scattering}.
A collection of inner symmetries may form a group $G$. In such a case
we say that $G$ acts on the net by inner symmetry.
Similarly one can consider an action of $V$ on Borchers triples. In this case, one says that $V$ implements
an inner symmetry if $\Ad V$ preserves $\M$ and commutes with $T$.
We say also that $G$ acts by inner symmetry when such $V$'s form a group.

\subsection{Conformal nets and Longo-Witten endomorphisms}
In some of our constructions, the main ingredients come from chiral conformal field theory.
Let us summarize here its operator-algebraic treatment (see also \cite{GF93, Longo08}).
An open, connected, non dense and non empty subset $I$ of $S^1$ is called an {\bf interval}
in $S^1$. We identify $\RR$ as a dense subset in $S^1$ by the stereographic projection.
The M\"obius group $\psl2r$ acts on $S^1 = \RR\cup\{\infty\}$ and under this identification it contains
translations and dilations of $\RR$.
A {\bf (M\"obius covariant) net (of von Neumann algebras) on $S^1$} is an assignment of
von Neumann algebras $\A_0(I)$ on a common Hilbert space $\H_0$ to intervals $I$
such that
\begin{enumerate}
\item[(1)] {\bf (Isotony)} If $I_1 \subset I_2$, then $\A_0(I_1) \subset \A_0(I_2)$.
\item[(2)] {\bf (Locality)} If $I_1$ and $I_2$ are disjoint, then $\A_0(I_1)$ and $\A_0(I_2)$ commute.
\item[(3)] {\bf (M\"obius covariance)} There is a (strongly continuous) unitary representation $U_0$ of
$\psl2r$ such that $U_0(g)\A_0(I)U_0(g)^* = \A_0(gI)$ for $g \in \psl2r$.
\item[(4)] {\bf (Positivity of energy)} The generator of the subgroup of translation in the representation
$U_0$ has positive spectrum.
\item[(5)] {\bf (Vacuum vector)} There is a (up to scalar) unique vector $\Omega_0$ which is invariant
under $U_0(g)$ and cyclic for $\A_0(I)$.
\end{enumerate}
As in two dimensions, we call $\A_0$ a net, however the actual object of interest is the triple
$(\A_0, U_0, \Omega_0)$. From these assumptions many properties automatically follow, among which
of importance in our application are
\begin{enumerate}
\item[(6)] {\bf (Irreducibility)} The von Neumann algebra $\bigvee_{I \subset S^1} \A_0(I)$ is
equal to $B(\H_0)$.
\item[(7)] {\bf (Haag duality on $S^1$)} It holds that $\A_0(I)' = \A_0(I')$, where $I'$ denotes the
interior of the complement of $I$ in $S^1$.
\item[(8)] {\bf (Bisognano-Wichmann property)} The modular group $\Delta_0^{it}$ of $\A_0(\RR_+)$
associated with $\Omega_0$ is equal to $U_0(\delta(-2\pi t))$, where $\delta$ is the dilation in
$\psl2r$.
\end{enumerate}

An {\bf inner symmetry} of a conformal net $\A_0$ is a collection of automorphisms of each
local algebra $\A_0(O)$ implemented by a common unitary operator $V_0$ 
which preserves the vacuum state $\<\Omega_0, \cdot\,\Omega_0\>$.
In conformal (M\"obius covariant) case, it automatically follows that $V_0$ commutes with $U_0(g)$ thanks to
Bisognano-Wichmann property.

A {\bf Longo-Witten endomorphism} of $\A_0$ is an endomorphism of $\A_0(\RR_+)$, implemented
by a unitary operator $V_0$ which commutes with $U_0(g)$, where $g$ is a translation. An inner symmetry
restricted to $\A_0(\RR_+)$ is a Longo-Witten endomorphism.

We will give a concrete example of conformal net and of Longo-Witten endomorphisms in Section \ref{u1},
which are not inner symmetries.
Note that the term ``Longo-Witten endomorphism'' in this paper does not mean
the concrete family found in \cite{LW11}. In fact, in \cite{LW11} it has been pointed out
that to a Longo-Witten endomorphism of a conformal net there is a time-translation
covariant net of von Neumann algebra on the half-plane in $\RR^2$.
A further family of examples has been found in \cite{BT12}.

\subsection{The massive scalar free field}\label{free}
Our main construction strategy is based on simpler examples with certain properties.
Let us quickly review the simplest quantum field.

The free field is constructed from an irreducible representation of the Poincar\'e group
through second quantization. We use the notation which is to some extent consistent with
\cite{LS12}. The one-particle Hilbert space is $\H_1 := L^2(\RR, d\theta)$ and for
the mass $m > 0$ the (proper orthochronous) Poincar\'e group acts by
$U_1(a,\lambda)\psi(\theta) = e^{ip(\theta)\cdot a}\psi(\theta-\lambda)$, where
$p(\theta) := (m\cosh(\theta), m\sinh(\theta))$ parametrizes the mass shell.
We introduce the (auxiliary) unsymmetrized Hilbert space $\H^\Sigma := \bigoplus \H_1^{\otimes n}$ and
the (physical) symmetrized Hilbert space $\Hr := \bigoplus Q_n \H_1^{\otimes n}$, where $Q_n$
is the projection onto the symmetric subspace.

Let us denote the $n$-th component of a vector $\Psi \in \Hr$ by $(\Psi)_n$.
For a vector $\psi \in \H_1$ and $\Psi$ which has only finitely many components,
the creation operator $b^\dagger$ is defined by 
$(b^\dagger(\psi)\Psi)_n = \sqrt{n}Q_n(\psi\otimes \Psi_{n-1})$. The annihilation operator
is the adjoint $b(\psi) = b^\dagger(\psi)^*$.
With this notation, $b^\dagger$ is linear and $b$ is antilinear with respect to $\psi$.
The free quantum field $\phi$ is now given by
\[
\phi(f) := b^\dagger(f^+) + b(J_1f^-), \;\;\;\;\;\;\;\;\;\; f^\pm(\theta) = \frac{1}{2\pi}\int d^2a f(a)e^{\pm ip(\theta)\cdot a},
\]
where $f$ is a test function in $\mathscr{S}(\RR^2)$ and $J_1\psi(\theta) = \overline{\psi(\theta)}$.
This field is local, in the sense that if
$f$ and $g$ have spacelike separated supports, then $\phi(f)$ and $\phi(g)$ commute on an
appropriate domain.

Finally we introduce the free field net. For each open region $O\subset \RR^2$,
we set
\[
\Ar(O) := \{e^{i\phi(f)}: \supp f \subset O\}''.
\]
We have the second quantized representation $\Ur := \Gamma(U_1)$ and the Fock vacuum
vector $\Omr \in \Hr$. This triple $(\Ar,\Ur,\Omr)$ is a two-dimensional
Haag-Kastler net and referred to as the {\bf free massive scalar net}.
The subscript r is intended for the {\it real} scalar field.
Furthermore, this net satisfies the modular nuclearity \cite[Section 4]{BL04}, a property which we will
explain in more detail in Section \ref{general}.

\subsection{Scalar factorizing S-matrix models}\label{lechner}
The free net has many important features, but is not interacting.
Lechner has constructed a large family of interacting nets \cite{Lechner08}.
Here we only briefly summarize the construction and their fundamental properties.

Let $S_2(\theta)$ be a bounded analytic function
on the strip $0 < \Im \theta < \pi$, continuous on the boundary, with the property
\[
S_2(\theta)^{-1} = \overline{S_2(\theta)} = S_2(-\theta) = S_2(\theta+i\pi)
\]
for $\theta \in \RR$. This is called the two-particle scattering function.

This time again the construction of the (wedge-local) field and the net is based on
the one-particle space $\H_1$ above. There is a representation of the symmetric group
$\mathfrak S_n$ on $\H_1^{\otimes n}$. For $\Psi_n \in \H_1^{\otimes n}$, the ``$S_2$-transposition''
is given by
\[
(D_{S_2,n}(\tau_j)\Psi)(\theta_1,\cdots, \theta_n) = S_2(\theta_{j+1}-\theta_j)\Psi(\theta_1,\cdots,\theta_{j+1},\theta_j,\cdots,\theta_n).
\]
for $\tau_j$ which transposes $j$ and $j+1$, and this generate a representation of
$\mathfrak S_n$.

Let $Q_{S_2,n}$ be the orthogonal projection onto the subspace of $\H_1^{\otimes n}$ invariant
under $\{D_{S_2,n}(\tau_j): 1 \le j \le n-1\}$.
The full Hilbert space is $\H_{S_2} := \bigoplus Q_{S_2,n} \H_1^{\otimes n}$ and the representation $U_1$
promotes to $U_{S_2}$ on $\H_{S_2}$ by the second quantization.
We define similarly the creation and annihilation operators
$(z_{S_2}^\dagger(\psi)\Phi)_n = \sqrt{n}Q_{S_2,n}(\psi\otimes \Phi_{n-1})$
and $z_{S_2}(\psi) = z_{S_2}^\dagger(\psi)^*$. 
The quantum field is defined accordingly by
\[
\phi_{S_2}(f) := z_{S_2}^\dagger(f^+) + z_{S_2}(J_1f^-), \;\;\;\;\;\;\;\;\;\;
f^\pm(\theta) = \frac{1}{2\pi}\int d^2a f(a)e^{\pm ip(\theta)\cdot a},
\]
but this time $\phi_{S_2}$ is only wedge-local, hence the net is defined through Borchers triple.

The von Neumann algebra $\M_{S_2}$ is defined by
\[
\M_{S_2} := \{e^{i\phi_{S_2}(f)}: \supp f \subset W_\L\}'
\]
(note that here we take only the single commutant and $f$ has support in $W_\L$,
while $\M_{S_2}$ corresponds to $W_\R$. By the wedge-duality, this is just a matter of convention
(see \cite{LS12})).
The triple $(\M_{S_2}, U_{S_2}, \Omega_{S_2})$ is a Borchers triple \cite{Lechner03}.
Furthermore, the modular nuclearity holds if $S_2$ fulfills a certain regularity condition
and $S_2(0) = -1$ \cite{Lechner08},
hence the triple is wedge-split in those cases.

\section{Borchers triples through inner symmetries}\label{inner}
In this Section we make the first step in our main construction in two dimensions.
We start with a Borchers triple with inner symmetry and construct new triples.
Note that, differently from our previous results \cite{Tanimoto12-2, BT12}, the new triples are
defined on a different Hilbert space although the formulae look very similar.

First we treat the case where a given triple admits an action of $S^1$ by inner symmetry.
Let us state a key lemma, which can be obtained as a special case of \cite[Lemma 4.1]{Tanimoto12-2}.
\begin{lemma}\label{lm:commutativity-s}
Let $\Mc$ be a von Neumann algebra and $\Qc$ a self-adjoint operator such that
$\Ad e^{i2\pi \k\Qc}(\Mc) = \Mc$ for any $\k\in \RR$.
Then, for $\k\in \RR$,
$x\otimes \1$ commutes with $\Ad e^{i2\pi\k \Qc\otimes \Qc}(x'\otimes \1)$
where $x\in \Mc$ and $x'\in\Mc'$.
\end{lemma}
\begin{proof}
For any $\k \in \RR$, $x$ commutes with
$\Ad e^{i2\pi\k\Qc}(x')$ by assumption, hence the commutativity lemma \cite[Lemma 4.1]{Tanimoto12-2}
applies.
\end{proof}

Note that our assumption is for every $\k \in \RR$, since in the partial
spectral decomposition in the proof of \cite[Lemma 4.1]{Tanimoto12-2},
there appears the action $\Ad e^{i2\pi q\k\Qc}$ on the left component, $q\in\RR$.
Actually, if the spectrum of $\Qc$ is containd in a subset $X \subset \RR$,
then it is enough to assume that $\Ad e^{i2\pi q\k\Qc}(\Mc') \subset \Mc'$ for any $q\in X$. We will treat concrete cases where $\sp \Qc \subset \ZZ$.
Here we put the factor $2\pi$ in order to keep the notations homogeneous
to the actions of $\ZZ_N$ considered later.

Now let $(\Mc, \Tc, \Omc)$ be a Borchers triple with an action of $S^1 = \RR/\ZZ$
by inner symmetry.
There is a unique unitary operator $\Vc(\k)$ which implements the inner symmetry by $S^1$.
One take the generator $\Qc$ such that $\Vc(\k) = e^{i2\pi \k\Qc}$.
It is clear that $\Qc$ commutes with the translation $\Tc$.
One sees that $\sp \Qc \subset \ZZ$ since $\1 = \Vc(1) = e^{i2\pi\Qc}$.

Now we turn to the construction of Borchers triples. Our objects act on the tensor product Hilbert
space $\Htc := \Hc\otimes \Hc$.
Let us denote $\Vtcs = e^{i 2\pi \k\Qc\otimes \Qc}$.
The tensor product representation $\Ttc(a) := \Tc(a)\otimes \Tc(a)$ has positive spectrum
and preserves the new vacuum $\Omtc := \Omc\otimes \Omc$.
The von Neumann algebra is given respectively by
\begin{eqnarray*}
\Mtcs &:=& \{x\otimes \1, \Ad \Vtcs (\1\otimes y): x,y\in \Mc\}'', \\
\end{eqnarray*}
Note the difference from our previous construction \cite{Tanimoto12-2}: first of all, this time the input
is the two-dimensional Borchers triple, while we used chiral components in \cite{Tanimoto12-2}.
Accordingly, the representation $\Ttc$ is now just the two copies of the given one. The formula for
$\Mtcs$ is also similar, but this time we take both $x$ and $y$ from the same algebra $\Mc$.
\begin{theorem}\label{th:borchers-inner}
The triple $(\Mtcs, \Ttc, \Omtc)$ is a Borchers triple for each $\k\in\RR$.
\end{theorem}
\begin{proof}
As in \cite[Theorem 4.17]{Tanimoto12-2}, the properties of $\Ttc$ and $\Omtc$ are readily checked.
As for the relation $\Ad \Ttc(a) (\Mtcs) \subset \Mtcs$ for $a \in W_\R$, one can prove this by noting
that $\Vtcs$ and $\Ttc(a)$ commute and the assumptions that $\Mc$ and $\Tc$ have the relation.

The cyclicity of $\Omtc$ for $\Mtcs$ follows from the cyclicity of $\Omc$ for $\Mc$
and the fact that $\Vtcs(\1\otimes y)\Omtc = (\1\otimes y)\Omtc$. To show the separating
property, we need again to prepare a sufficiently big algebra in the commutant:
\[
\Mtcs^1 := \{\Ad \Vtc(x'\otimes \1), \1\otimes y': x', y' \in \Mc'\}'' .
\]
It is easy to see that $\Omtc$ is cyclic for $\Mtcs^1$. In order to see that $\Omtc$ is separating
for $\Mtcs$, it is enough to prove that $\Mtcs$ and $\Mtcs^1$ commute.
This follows from Lemma \ref{lm:commutativity-s}, as in \cite[Theorem 4.2]{Tanimoto12-2}.
\end{proof}

A similar construction is possible if a Borchers triple admits an
action of the finite group $\ZZ_N$ by inner symmetry.
For a $\ZZ_N$-action we take $\Vc$ such that $\Vc^k$ implements the inner symmetry for $k\in\ZZ_N$.
One can choose $\Qc$ such that $\Vc(\k) = e^{i2\pi \k\Qc}$ or $\Vc = e^{i\frac{2\pi}{N}\Qc}$,
respectively. We remark that we can concretely choose $\Qc$ as follows: for each integer $j \in [0,N-1]$
we put $\hat \Vc(j) := \frac{1}{N}\sum_k e^{-i\frac{2\pi jk}{N}}\Vc(k)$. It is clear that $\hat \Vc(j)$
is an orthogonal projection and $\hat \Vc(j) \hat \Vc(l) = 0$ if $j\not\equiv l \mod N$. In other words,
$k\mapsto \Vc^k$ is a representation of $\ZZ_N$, whose dual group is also $\ZZ_N$ and we extract
the spectral projections. Then we can take $\Qc := \sum_j j \hat \Vc(j)$.
By this construction, it is clear that $\Qc$ commutes with the translation $\Tc$.
An operator $\Qc$ with the properties specified above is not unique since one can amplify each spectral
component by $N$, but if we consider the operator $e^{i\frac{2\pi k}{N}\Qc\otimes \Qc}$,
it does not depend on the choice of $\Qc$
(see also the remark after \cite[Theorem 4.17]{Tanimoto12-2}).  
One sees that $\sp \Qc \subset \ZZ$ since $\1 = \Vc^N = e^{i2\pi \Qc}$ for a $\ZZ_N$-action (or directly by our choice).

Analogously as Lemma \ref{lm:commutativity-s} one can prove the following.
Here, the assumption is for all intergers and so is the result.
\begin{lemma}\label{lm:commutativity}
Let $N$ be an integer, $\Mc$ a von Neumann algebra and $\Qc$ a self-adjoint operator such that
$\sp \Qc \subset \ZZ$ and $\Ad e^{i\frac{2\pi k}{N} \Qc}(\Mc) = \Mc$ for any $k\in \ZZ$.
Then, for $k\in \ZZ$,
$x\otimes \1$ commutes with $\Ad e^{i\frac{2\pi k}{N} \Qc\otimes \Qc}(x'\otimes \1)$
where $x\in \Mc$ and $x'\in\Mc'$.
\end{lemma}

The construction of Borchers triples is also parallel. We take the tensor product Hilbert
space $\Htc := \Hc\otimes \Hc$, the tensor product representation
$\Ttc(a) := \Tc(a)\otimes \Tc(a)$ and the vacuum $\Omtc := \Omc\otimes \Omc$.
We denote $\Vtc = e^{i\frac{2\pi k}{N}\Qc\otimes \Qc}$ and define
\begin{eqnarray*}
\Mtc &:=& \{x\otimes \1, \Ad \Vtc (\1\otimes y): x,y\in \Mc\}''.
\end{eqnarray*}
\begin{theorem}\label{th:borchers-inner-z}
The triple $(\Mtc, \Ttc, \Omtc)$ is a Borchers triple for each $k\in\ZZ_N$.
\end{theorem}

At the end of this Section, we remark that the existence of inner symmetry is not at all
exceptional. Indeed, if $(\Mc, \Tc, \Omc)$ is a (wedge-split) Borchers triple,
then the tensor product $(\Mc\otimes \Mc, \Tc\otimes \Tc, \Omc\otimes \Omc)$
has the flip automorphism which commutes with $\Tc\otimes \Tc$ and preserves
$\Omc\otimes\Omc$, hence it is an inner symmetry $\ZZ_2$.
We present more examples in Section \ref{scattering}.

\section{General arguments for strict locality}\label{general}
We are now concerned with the main problem in the construction of nets through Borchers triples:
the strict locality of the models in Section \ref{inner}.
The key argument is the wedge-split property. 

In general, if one has an inclusion of von Neumann algebras $\N \subset \M$, this is said to be
{\bf split} if there is a type I factor $\Rc$ such that $\N \subset \Rc \subset \M$.
Furthermore, if $\N\subset \M$ is an inclusion of factors and $\Omega$ is cyclic and
separating for $\M$, then the nuclearity of the map $\N \ni x \longmapsto \Delta^{\frac{1}{4}}x\Omega$
implies the split property, where $\Delta$ is the modular operator for $\M$ with respect to $\Omega$
\cite[Proposition 2.3]{BDL90}.

We say that a Borchers triple $(\M, T, \Omega)$ is {\bf wedge-split} if $\Ad T(a)(\M) \subset \M$
is split for any $a\in W_\R$. Note that $W_\R$ is defined as the {\em open} wedge
and split inclusion for lightlike translation is not required.
Wedge-split property implies that the inclusion $\Ad T(a)(\M) \subset \M$ is unitarily equivalent
to $(\M_2\otimes \CC\1) \subset (B(\K_1)\otimes \M_1)$, where $\M_1$ and $\M_2$ cannot be trivial
since $\M$ is of type III, hence
the intersection $\M\cap \Ad T(a)(\M)'$ for $a\in W_\R$ is unitarily equivalent to
$\M_2'\otimes \M_1$ which is nontrivial \cite{BL04}.
Let $\Delta$ be the modular operator for $\M$ with respect to $\Omega$.
One says that an inclusion $\N\subset \M$ satisfies {\bf modular nuclearity} if
the map $\N \ni x \longmapsto \Delta^{\frac{1}{4}}x\Omega$ is nuclear.
A Borchers triple $(\M,T,\Omega)$ is said to satisfy modular nuclearity if
the inclusion $\Ad T(\M) \subset \M$ has modular nuclearity for any $a\in W_\R$.
From the above remark, modular nuclearity implies wedge-split property.
The strict locality of the models in \cite{Lechner08} was proved through modular nuclearity.

Actually, wedge-split inclusion is sufficient for the strict locality \cite[Theorem 2.5]{Lechner08}.
\begin{theorem}[Lechner]\label{th:lechner}
If a Borchers triple $(\M,T,\Omega)$ is wedge-split, then it is strictly local.
\end{theorem}
This theorem is stated with the assumption of modular nuclearity, however, the actual proof
depends only on wedge-split property.

Let us recall our main construction strategy (Section \ref{inner}): starting with a given Borchers triple
$(\Mc,\Tc,\Omc)$ with inner symmetry $\Ad \Vc(\k)$ of $S^1$,
we constructed a new Borchers triple $(\Mtcs,\Ttc,\Omtc)$ on the tensor
product Hilbert space. As we will see later (Section \ref{scattering}), this gives a nontrivial scattering
even if the initial triple comes from the free field. This can be generalized to the following program:
take a Borchers triple with a good property (either wedge-split property or modular nuclearity)
and prove that the construction of Section \ref{inner} leads again to Borchers triples with the same
property.

We carry out this program in two ways. First we present a proof through wedge-split property and
then we use modular nuclearity.
Note that strict locality is not ``good enough'' for this program. We will exhibit examples
of strictly local Borchers triples for which the new triples constructed as in Section \ref{inner}
violate strict locality.
Another remark is that our arguments in Section \ref{split} through wedge-split property are valid
for both $S^1$- and $\ZZ_N$-actions,
but those in Section \ref{nuclearity} through modular nuclearity apply so far only to finite cyclic group actions.
Already for the simplest compact group $S^1$ the proofs break down, as we will see.

\subsection{Common arguments}\label{common}
Let us collect some facts which will be commonly used in the following
(see also \cite[Section 4.4.1]{Tanimoto12-2}).
In this Subsection, we do not use the properties of Borchers triple.

In the following, we consider only an action of $S^1$. The case of $\ZZ_N$ is analogous.
Let $\Mc$ be a von Neumann algebra, $\Omc$ a cyclic separating vector for $\Mc$ and
$\Vc(\k)$ be a unitary representation of $S^1$ which preserves $\Omc$ such that
$\acs := \Ad \Vc(\k)$ is an automorphism of $\Mc$ and $\Vc(1) = \1$.
In other words, there is an action
$\ac$ of $S^1$ on $\Mc$ which preserves the state $\<\Omc,\cdot\,\Omc\>$.

\subsubsection*{Discrete Fourier expansion on von Neumann algebras}
For any element $x\in\Mc$ and $l \in \ZZ$, we define
\[
x_l := \int_0^1 d\k\, \Ad \Vc(\k)(x)e^{-i2\pi l\k},
\]
It holds that $x = \sum_{l \in \ZZ} x_l$ and $\Ad \Vc(\k)(x_l) = \acs(x_l) = e^{i2\pi l\k}x_l$.
The sum is $*$-strongly convergent.
Let us denote by $\Mc^{\ac}(l)$ the set of elements $x$ of $\Mc$ such that
$\acs(x) = e^{i2\pi l\k}x$. Then it holds that $x_l \in \Mc^{\ac}(l)$ and
$\Mc^{\ac}(l_1)\cdot\Mc^{\ac}(l_2) \subset \Mc^{\ac}(l_1+l_2)$.

These are well-known facts, however, if necessary the reader is referred to
\cite[Proposition 4.9]{Tanimoto12-2} for a proof (for the case of $\ZZ_N$, the proof is
easy to adapt).

Next we consider the twisting. As in Section \ref{inner}, let $\Qc$ be a self-adjoint
operator such that $\Vc(\k) = e^{i2\pi \k\Qc}$.
Then for $y \in \Mc^{\ac}(m)$, it holds that
\[
\Ad \Vtcs(\1\otimes y) = \Ad e^{i2\pi \k\Qc\otimes \Qc}(\1\otimes y) = e^{i2\pi m\k\Qc}\otimes y
= \Vc(m\k)\otimes y,
\]
where $\Vtcs = e^{i2\pi \k\Qc\otimes \Qc}$ as in Section \ref{inner}.
For the proof, see \cite[Lemma 4.10]{Tanimoto12-2}. The proof is written again for
an action of $S^1$, however, the adaptation is easy for $\ZZ_N$. Note that one has to consider
spectral subspaces of $\Vc$ parametrized by $\ZZ_N$. The identification
$\ZZ_N = \ZZ / N\ZZ$ should be always kept in mind and numbers, e.g.,  $e^{i\frac{2\pi k}{N}}$
are well-defined for $k\in\ZZ_N$.

Finally we remark that $\Mtcs$ is the $*$-strong closure of the linear span of
the elements of the form $x_l\Vc(m\k)\otimes y_m$, where $x_l\in \Mc^{\ac}(l)$ and
$y_m\in\Mc^{\ac}(m)$.
Furthermore,  any element $\widetilde{x} \in \Mtcs$ can be decomposed as follows:
\[
 \widetilde{x} = \sum_{l,m \in \ZZ} \widetilde{x}_{l,m} (\Vc(m\k)\otimes \1),
 \,\,\,\,\,\,\, \widetilde{x}_{l,m} \in \Mc\otimes \Mc,
\]
and the sum is $*$-strongly convergent.
This decomposition corresponds to the discrete Fourier expansion with respect
to the action $(\Ad \Vc)\otimes(\Ad \Vc)$ of the group $S^1\times S^1$.
The group $S^1\times S^1$ acts also on $\Mc\otimes \Mc$ and there is a decomposition
into components as well.
For a corresponding proof, see again \cite[Lemma 4.14]{Tanimoto12-2}.

\subsubsection*{Modular operator and inner symmetry}
Let $\Dc$ and $\Jc$ be the modular operator and the modular conjugation
for $\Mc$ with respect to $\Omc$. The implementing unitary $\Vc$ of $\ac$ and
$\Dc, \Jc$ commute. 
The modular operator and the modular conjugation for $\Mc\otimes\Mc$ with
respect to $\Omtc := \Omc\otimes\Omc$ are $\Dc\otimes \Dc$ and
$\Jc\otimes \Jc$, respectively. Now we can determine the modular objects
for $\Mtcs = \Mc\otimes \1 \vee \Ad \Vtcs (\1\otimes \Mc)$.
\begin{proposition}\label{pr:modular}
The operators $\Dc\otimes\Dc$ is the modular operator for $\Mtc$ with respect to
$\Omtc$.
\end{proposition}
\begin{proof}
We will give two proofs. The first one is a direct calculation and the second one uses
the so-called KMS condition.

A direct calculation goes as follows. We will see that actually the modular conjugation
is $\Vtcs (\Jc\otimes \Jc)$. Since we have already the candidates for the
modular objects, we only have to check their actions. First we take $x_l \in \Mc^{\ac}(l)$
and $y_m \in \Mc^{\ac}(m)$. The operator $x_l\Vc(m\k)\otimes y_m$ belongs to
$\Mtcs$ and it holds that
\begin{eqnarray*}
\Vtcs(\Jc\otimes \Jc)(\Dc\otimes \Dc) (x_l\Vc(m\k)\otimes y_m) \Omtc
&=& \Vtcs (S_c\otimes S_c) (x_l\Omc\otimes y_m\Omc) \\
&=& \Vtcs (x_l^*\Omc\otimes y_m^*\Omc) \\
&=& e^{i2\pi lm\k}(x_l^*\Omc\otimes y_m^*\Omc) \\
&=& (\a_{\mathrm{c},-m\k}(x_l^*)\Omc)\otimes (y_m^*\Omc) \\
&=& (\Vc(-m\k) x_l^*\Omc) \otimes (y_m^*\Omc) \\
&=& (x_l\Vc(m\k)\otimes y_m)^*(\Omc\otimes\Omc),
\end{eqnarray*}
namely, the operator $\Vtcs(\Jc\otimes \Jc)(\Dc\otimes \Dc)$ acts correctly as
the modular involution. Since $\Mtcs$ is the $*$-strong closure of
the linear span of such elements,
the modular involution actually coincides
with $\Vtcs(\Jc\otimes \Jc)(\Dc\otimes \Dc)$. Then the conclusion follows
immediately from the uniqueness of the polar decomposition.

In the second proof, we use the uniqueness of the modular automorphism group
with respect to the KMS condition. In general, for a von Neumann algebra $\M$
and a state $\p$, if there is a one-parameter automorphisms $\sigma_t$ which
satisfies the KMS condition, namely if for any pair $x,y \in \M$ there is an analytic function
$f(t)$ on the strip $0 < \Im t < 1$ and continuous on the boundary such that
$f(t) = \p(\sigma_t(x)y)$ and $f(t+i) = \p(y\sigma_t(x))$ for $t\in\RR$,
then $\sigma_t$ is the modular automorphism \cite[Theorem VIII.1.2]{TakesakiII}.
In our case, the state is the vacuum
$\<\Omtc, \cdot\,\Omtc\>$. This is the KMS state on the von Neumann algebra
$\Mc\otimes\Mc$ with respect to $\sigma^{\Omtc}_t = \Ad (\Dc^{it}\otimes \Dc^{it})$.
Hence for any pair $\widetilde{x},\widetilde{y}\in \Mc\otimes\Mc$
of elements there is an analytic function as above.
Actually the decomposition with respect to $\ac\otimes\ac$ commutes with $\sigma^{\Omtc}$
(see the remark about discrete Fourier expansion) and it holds that
\begin{eqnarray*}
\<\Omtc, \sigma^{\Omtc}_t(\widetilde{x})\widetilde{y}\Omtc\>
&=& \sum_{l,m}\<\Omtc, \sigma^{\Omtc}_t(\widetilde{x}_{l,m})\widetilde{y}_{-l,-m}\Omtc\>, \\
\<\Omtc, \widetilde{y}\sigma^{\Omtc}_t(\widetilde{x})\Omtc\>
&=& \sum_{l,m}\<\Omtc, \widetilde{y}_{-l,-m}\sigma^{\Omtc}_t(\widetilde{x}_{l,m})\Omtc\>
\end{eqnarray*}
by the orthogonality of the vacuum acted on by the component $\widetilde{x}_{l,m}$ etc.
Next, by considering the pair $\widetilde{x}_{l,m}, \widetilde{y}_{-l,-m}$ and the KMS condition for $\Mc\otimes\Mc$
there is an analytic function $\widetilde{f}_{l,m}(t) = \<\Omtc, \sigma^{\Omtc}_t(\widetilde{x}_{l,m})\widetilde{y}_{-l,-m}\Omtc\>$,
$\widetilde{f}_{l,m}(t+i) = \<\Omtc, \widetilde{y}_{-l,-m}\sigma^{\Omtc}_t(\widetilde{x}_{l,m})\Omtc\>$ (of course the dependence of
$\widetilde{f}_{l,m}$ on $\widetilde{x}$ and $\widetilde{y}$ is implicit).
Let us turn to $\Mtcs$. By the same argument as above, an inner product decomposes into
suitable combination and the decomposition is compatible with the action of
$\sigma^{\Omtc}_t = \Ad (\Dc^{it}\otimes\Dc^{it})$.
For a pair of elements $\widetilde{x}_{l,m}(e^{i2\pi m\k\Qc}\otimes \1)$
and $\widetilde{y}_{-l,-m}(e^{-i2\pi m\k\Qc}\otimes \1)$,
where $\widetilde{x}_{l,m}, \widetilde{y}_{-l-m} \in \Mc\otimes\Mc$, we have
\begin{eqnarray*}
\<\Omtc, \sigma^{\Omtc}_t(\widetilde{x}_{l,m}(e^{i2\pi m\k\Qc}\otimes \1))\widetilde{y}_{-l,-m}(e^{-i2\pi m\k\Qc}\otimes \1)\Omtc\>
&=& e^{-i2\pi lm\k}\<\Omtc, \sigma^{\Omtc}_t(\widetilde{x}_{l,m})\widetilde{y}_{-l,-m}\Omtc\> \\
\<\Omtc, \widetilde{y}_{-l,-m}(e^{-i2\pi m\k\Qc}\otimes \1)\sigma^{\Omtc}_t(\widetilde{x}_{l,m}(e^{i2\pi m\k\Qc}\otimes \1))\Omtc\>
&=& e^{-i2\pi lm\k}\<\Omtc, \widetilde{y}_{-l,-m}\sigma^{\Omtc}_t(\widetilde{x}_{l,m})\Omtc\>
\end{eqnarray*}
One observes that the right hand sides are equal to $\widetilde{f}_{l,m}$ up to the constant $e^{-i2\pi lm\k}$.
In other words, the KMS condition is satisfied for the pair with respect to $\sigma^{\Omtc}$.
Therefore for an arbitrary linear combination of such components the KMS condition holds as well,
thanks to the decomposition of the inner product into $(l,m)$-components.
Such linear combinations is $*$-strongly dense in $\Mtcs$, hence the above-cited
uniqueness theorem applies to see that $\sigma^{\Omtc}_t = \Ad (\Dc^{it}\otimes\Dc^{it})$
is the modular automorphism of $\Mtcs$ with respect to $\<\Omtc, \cdot\,\Omtc\>$
(indeed, the KMS condition on a $*$-strongly dense subalgebra is enough by
\cite[Proposition 5.7]{BR2}. One should note that in \cite{BR2} the KMS condition is
defined on a dense set of analytic elements).
Since $\Dc\otimes\Dc$ preserves the vacuum vector $\Omc\otimes\Omc$,
it must coincide with the modular operator.
\end{proof}

It is interesting to compare our proof with \cite{BL04, Lechner08} where the modular
objects were calculated through unbounded operators affiliated to the von Neumann algebras.
One recalls also that in our previous work the modular objects were indirectly determined
by scattering theory \cite[Section 3]{Tanimoto12-2}.

\begin{proposition}\label{pr:commutant}
It holds that
$(\Mtcs)' = \left(\Ad \Vtc(\Mc'\otimes\1)\vee\1\otimes\Mc'\right) =: \Mtcs^1$.
\end{proposition}
\begin{proof}
It can be observed that the right hand side commute with
$\Mtcs$ by the same argument as in Theorem \ref{th:borchers-inner}, so we have
$\Mtcs^1 \subset (\Mtcs)'$.
Since we already know that the modular operator for $(\Mtcs)'$
with respect to $\Omc\otimes\Omc$ is $\Dc^{-1}\otimes\Dc^{-1}$ from Proposition \ref{pr:modular},
it is immediate to see that $\Mtcs^1$ is globally invariant under the modular group of
$(\Mtcs)'$ with respect to $\Omc\otimes\Omc$. Then the two algebras coincide
since $\Omc\otimes\Omc$ is cyclic for $\Mtcs^1$ by a standard application of Takesaki's theorem
\cite[Theorem IX.4.2]{TakesakiII} (see \cite[Theorem A.1]{Tanimoto12-1}).
\end{proof}
Variations of Propositions \ref{pr:modular} and \ref{pr:commutant} hold for an action of $\ZZ_N$ with
trivial changes.

\subsection{Proof through wedge-split property}\label{split}
Here we take as the starting point a Borchers triple $(\Mc,\Tc,\Omc)$ which is wedge-split.
As before, we assume that there is an action of $S^1$ implemented by $\Vc(\k)$.
With this action and for $\k \in \RR$ we can construct a Borchers triple $(\Mtcs, \Ttc, \Omtc)$
as in Section \ref{inner}. We are going to prove that this is again wedge-split, hence strictly local.

Recall that \cite{DL84} for a split inclusion $(\Nc \subset \Mc, \Omc)$ equipped with a cyclic separating
vector $\Omc$ for $\Mc, \Nc$ and $\Mc\cap \Nc'$ (where $\Nc$ will be $\Ad \Tc(a)(\Mc)$ for some $a \in W_\R$),
there is a {\em canonical} type I factor $\Rcc$
such that $\Nc \subset \Rcc \subset \Mc$. Moreover, $\Rcc$ is given by the formula
$\Rcc = \Nc \vee J\Nc J = \Mc \cap J\Mc J$, where $J$ is the modular conjugation for $\Mc\cap \Nc '$ with
respect to $\Omc$. 
If $\Vc$ is a unitary operator which preserves $\Omc$
and $\Ad \Vc$ preserves both $\Nc$ and $\Mc$, then $\Rcc$ is preserved under $\Ad \Vc$ as well.

\begin{lemma}\label{lm:factoriality}
The von Neumann algebra $\Rtccs := \Rcc\otimes \1\vee\Ad \Vtcs(\1\otimes \Rcc)$
is a factor.
\end{lemma}
\begin{proof}
Since we know that $(\Rtccs)' = \Ad \Vtcs(\Rcc'\otimes\1)\vee\1\otimes\Rcc'$ by Proposition \ref{pr:commutant}
(in Section \ref{common} we did not assume that the von Neumann algebra comes from a Borchers triple),
we only have to show that $\Rtccs\vee(\Rtcc)' = \Rcc\otimes\Rcc'\vee\Ad \Vtcs(\Rcc'\otimes\Rcc)$ is equal to
$B(\Hc\otimes\Hc)$.

We show that $\Vc(\k)\otimes \1$ and $\1\otimes \Vc(\k)$ are contained in $\Rtccs\vee(\Rtccs)'$
for each $\k\in\RR$.
Indeed, $\Vc(\k)\otimes\1$ implements an automorphism on $\Rcc\otimes \1$
and since the latter is a type I factor, there is a unitary $u \in \Rcc$ which implements
the same automorphism. At the same time $\Vc(\k)\otimes \1$ implements an automorphism
of $\Rcc'\otimes\1$, and there is an implementing unitary $u' \in \Rcc'$.
Then $uu'$ implements $\Ad \Vc(\k)$ on $\B(\Hc) = \Rcc\vee\Rcc'$, hence $uu'\Vc(\k)^*$ must be
a scalar. We may assume $uu' = \Vc(\k)$. Furthermore, obviously $\Ad \Vc(\k)(u') = \Ad u'(u') = u'$,
namely $u'$ is fixed under the automorphism $\Ad \Vc(\k)$, hence $\Ad \Vtcs(u'\otimes \1) = u'\otimes \1$.
This implies that $u'\otimes \1 \in (\Rtccs)'$ hence $uu'\otimes \1 = \Vc(\k)\otimes \1$ is in
$\Rtccs \vee (\Rtccs)'$. An analogous proof works for $\1\otimes \Vc(\k)$.

The rest is easy since $\Vtcs$ is obtained from the functional calculus of $\Vc(\k)\otimes\1$ and
$\1\otimes \Vc(\k)$.
\end{proof}
This Lemma actually works even for an action of $\RR$. In contrast, we need
the periodicity of the action in the following.

\begin{theorem}\label{th:type1}
The von Neumann algebra $\Rtccs := \Rcc\otimes \1\vee\Ad \Vtcs(\1\otimes \Rcc)$
is a type I factor.
\end{theorem}
\begin{proof}
We have seen that $\Rtccs$ is a factor. What remains is to show that $\Rtccs$ contains
a minimal projection. As in Lemma \ref{lm:factoriality} we take implementing unitary
$u(\k) \in \Rcc$ for $\Ad \Vc(\k)$, this time indicating the dependence on $\k$.
By a classical result by Bargmann, we may assume that $u(\k)$ is a one-parameter
group of unitaries in $\Rcc$ \cite[Theorem 1.1, Lemma 4.3]{Bargmann54}, where our group is $S^1$ hence
a one-parameter group.

Thus the spectrum of the representation $u$ is discrete and contained in $\ZZ$
(by multiplying a one-parameter phase),
hence the fixed point $\Rcc^{\ac} = \Rcc\cap\{u(\k): \k\in\RR\}'$ is
the relative commutant in $\Rcc$ of the spectral projections of the representation $u$.
Hence it is the direct sum of type I factors and admits a minimal projection.

Let us take such a minimal projection $p$ in $\Rcc^{\ac}$.
In general, if $y \in \Rcc^{\ac}$, then $\Ad \Vtcs(\1\otimes y) = \1\otimes y$ and hence
the inclusion $\Rcc\otimes\Rcc^{\ac} \subset \Rtccs$ holds.
Obviously $p$ must be a subprojection of
a spectral projector of $u$ corresponding to a discrete eigenvalue.
Furthermore, it is immediate that $p\otimes p$ is a minimal projection in $\Rcc\otimes \Rcc^{\ac}$.
We claim that it is a minimal projection in $\Rtccs$.
In fact, suppose that $\widetilde{q}\le p\otimes p$ in $\Rtccs$. Then $\widetilde{q}$ and $\1\otimes u(\k)$ commute
since the spectral projections of the representation $\1\otimes u$ either contains or disjoint from $p\otimes p$,
hence so does it hold with $\widetilde{q}$. This implies that $\widetilde{q}$ belongs to $\Rcc\otimes \Rcc^{\ac}$
(consider the discrete Fourier expansion $\widetilde{q} = \sum \widetilde q_{lm}(\Vc(m\k)\otimes \1)$, then
it is possible that $\Ad (1\otimes u(\k))(\widetilde q) = \widetilde q$ only if $\widetilde q_{lm}$ vanish except
$m = 0$),
then $\widetilde{q} = p\otimes p$ by the minimality of $p\otimes p$ in this restricted algebra.
This is the minimality of $p\otimes p$ in $\Rtccs$.
\end{proof}

We remark that the intermediate type I factor constructed here is not the canonical
one of \cite{DL84}. An explicit formula for it involves the modular conjugation of
the relative commutant, which is only abstractly determined.

The proof of Theorem \ref{th:type1} can be easily adapted to the case of an action
of $\ZZ_N$.

In terms of wedge-split Borchers triple, we apply Theorem \ref{th:type1} with $\Nc = \Ad \Tc(a)(\Mc)$
to see the split inclusion
\[
 \Nc\otimes\1\vee \Ad \Vtcs(\1\otimes\Nc) \subset \Rcc\otimes\1\vee \Ad \Vtcs(\1\otimes\Rcc)
\subset \Mc\otimes\1\vee \Ad \Vtcs(\1\otimes\Mc),
\]
where the middle element is $\Rtccs$ and the last element is $\Mtcs$. Since $\Ttc$ and $\Vtcs$
commute, $\Ad \Ttc(a)(\Mtcs)\subset \Mtcs$ is split with an intermediate type I factor $\Rtccs$
(which implicitly depends on $a$).
Then we have the following with the help of Theorem \ref{th:lechner}.
\begin{theorem}\label{th:strict-split}
If a given Borchers triple $(\Mc, \Tc, \Omc)$ with an action of
$S^1$ by inner symmetry is wedge-split, then
the Borchers triple $(\Mtcs, \Ttc, \Omtc)$ is wedge-split, hence strictly local.
\end{theorem}

One can repeat parallel proofs for an action of $\ZZ_N$ to obtain
the same result.

\subsection{Proof through modular nuclearity}\label{nuclearity}
Let us give another proof of strict locality, based on modular nuclearity.
Here we have to restrict ourselves to the case of an action of $\ZZ_N$.
Of course one can take an arbitrary finite subgroup $\ZZ_N$ in $S^1$, hence
physically this should not be considered as an essential restriction.
Let $(\Mtc,\Ttc,\Omtc)$ again be constructed as in Section \ref{inner}
for a fixed $k\in\ZZ_N$.
Again, thanks to Theorem \ref{th:lechner}, it is enough to prove
the modular nuclearity of the new triple.

Let us start with a trivial observation.
\begin{lemma}\label{lm:nuclearity-product}
Let $(\Mc, \Tc, \Omc)$ be a Borchers triple with modular nuclearity.
Then the triple $(\Mc\otimes \Mc, \Ttc, \Omtc)$ has modular nuclearity.
\end{lemma}
\begin{proof}
The modular objects are the tensor products as well, hence the nuclearity norm simply gets squared.
\end{proof}

\begin{theorem}\label{th:strict-nuclear}
Let $(\Mc, \Tc, \Omc)$ be a Borchers triple with modular nuclearity.
Then the triple $(\Mtc, \Ttc, \Omtc)$ has modular nuclearity.
\end{theorem}
\begin{proof}
For a fixed $a \in W_\R$, we have to show that the inclusion $\Ad \Ttc(a) (\Mtc) \subset \Mtc$
has modular nuclearity with respect to $\Omtc$. By assumption and Lemma \ref{lm:nuclearity-product}
we know that $(\Ad \Ttc(a) (\Mc\otimes\Mc) \subset \Mc\otimes\Mc)$ has modular nuclearity.
This means that the map
\[
\Mc\otimes\Mc \ni \widetilde x \longmapsto \left(\Dc^{\frac{1}{4}}\otimes\Dc^{\frac{1}{4}}\right)\widetilde{x}(\Omc\otimes\Omc)
\]
is a nuclear map.

Let us consider an element $\widetilde{x}$ of $\Mc\otimes\Mc$. As we have seen in the end of Section \ref{common},
(before Proposition \ref{pr:commutant}) it can be decomposed as $\widetilde{x} = \sum_{l,m} \widetilde{x}_{l,m}$,
and $\widetilde{x}_{l,m}$ are the matrix components with respect to the grading
given by the $\ZZ_N$-action and each $\widetilde{x}_{l,m}$ is in $\Mc\otimes\Mc$.
For such $\widetilde{x}$, we define an element $\t_k(\widetilde{x}) \in \Mtc$ as follows:
\[
\t_k(\widetilde{x}) = \sum_{l,m} \widetilde{x}_{l,m}(\Vc^{km}\otimes \1).
\]
This is a finite sum, thus there is no problem of convergence.
This map $\t_k$ is onto, as any element in $\Mtc$ in decomposed
in the form above. It is important to observe that the action on the vector $\Omtc$ is
unchanged. In other words, it holds that $\widetilde{x}\Omtc = \t_k(\widetilde{x})\Omtc$.

Now, we know that the modular operator of $\Mtc$ with respect to $\Omtc$ is
$\Dc\otimes \Dc$, hence the map in question is
$\t_k(\widetilde{x})\longmapsto (\Dc^{\frac{1}{4}}\otimes\Dc^{\frac{1}{4}})\t_k(\widetilde{x})\Omtc
= (\Dc^{\frac{1}{4}}\otimes\Dc^{\frac{1}{4}}) \widetilde{x}\Omtc$.
By assumption and Lemma \ref{lm:nuclearity-product} we know that
the map $\widetilde{x} \longmapsto (\Dc^{\frac{1}{4}}\otimes\Dc^{\frac{1}{4}})\t_k(\widetilde{x})\Omtc$
is nuclear.
\begin{center}
\begin{tikzpicture}[descr/.style={fill=white,inner sep=2.5pt}]
\matrix (m) [matrix of math nodes, row sep=3em,
column sep=3em]
{ \widetilde{x} & & \widetilde{x}\Omtc = \t_k(\widetilde{x}) \Omtc & (\Dc^{\frac{1}{4}}\otimes\Dc^{\frac{1}{4}})\widetilde{x}\Omtc\\
\t_k(\widetilde{x}) & & &\\ };
\path[|->,font=\normalsize]
(m-1-1) edge  (m-1-3)
(m-1-3) edge  (m-1-4)
(m-1-1) edge  (m-2-1)
(m-2-1) edge  (m-1-3.south);
\end{tikzpicture}
\end{center}

In the diagram above, the left triangle commutes and the straight line above is nuclear.
Since a composition of a bounded linear map and a nuclear map is again nuclear,
we only have to show that $\t_k^{-1}$ is bounded.
For this purpose, let us recall how explicitly $\{\widetilde{x}_{l,m}\}$ are obtained
(see \cite[Proposition 4.9]{Tanimoto12-2} for a corresponding formula for $S^1$-action):
\[
\widetilde{x}_{l,m} = \frac{1}{N^2}\sum_{j_1,j_2} e^{-i\frac{2\pi (j_1l +j_2m)}{N}} \Ad (\Vc^{j_1}\otimes \Vc^{j_2})(\widetilde{x}).
\]
Correspondingly for the decomposition of $\t_k(\widetilde x)$, we have
\[
\widetilde{x}_{l,m}(\Vc^{km}\otimes\1) = \frac{1}{N^2}\sum_{j_1,j_2} e^{-i\frac{2\pi (j_1l +j_2m)}{N}}
 \Ad (\Vc^{j_1}\otimes \Vc^{j_2})(\t_k(\widetilde{x})).
\]
In particular, we see $\|\widetilde{x}_{l,m}\| \le \|\t_k(\widetilde{x})\|$
and hence $\|\widetilde{x} \| = \|\sum_{l,m} \widetilde{x}_{l,m}\| \le N^2 \|\t_k(\widetilde{x})\|$.
In other words, $\t_k^{-1}$ is bounded by $N^2$.
\end{proof}

Note that here the finiteness of $N$ is crucial. The author does not know if
the same holds for the action of e.g. $S^1$, although one can always take an arbitrary finite subgroup
$\ZZ_N$ of $S^1$.

\subsection{Counterexamples from massless case}
In previous Sections we started with a strictly local Borchers triple, constructed new triples and
proved strict locality.
Although our method may seem quite simple, it is neither trivial nor
purely group-theoretic.
In order to see this, we take the massless models, for which one can construct
Borchers triples but strict locality fails.
See also \cite[Section 4.4]{Tanimoto12-2}.

Let $(\A_0, T_0, \Omega_0)$ be any strongly additive conformal (diffeomorphism covariant)
net with an action of $\ZZ_N$ by inner symmetry (there are many such nets, e.g.\! the tensor product of
the $\u1$-current or the loop group nets). A two-dimensional massless net can be constructed by tensor product:
$\A(I_+\times I_-) := \A_0(I_+)\otimes \A_0(I_-), T(t_+, t_-) := T_0(t_+)\otimes T_0(t_-),
\Omega := \Omega_0\otimes\Omega_0$. On the net $\A$ there is an action of $\ZZ_N$
by inner symmetry which acts only on the left component of $\A_0(I_+)\otimes\A_0(I_-)$. Let $Q$ be the
generator with $\sp Q \subset \ZZ$, which is of the form $Q_0\otimes\1$.
With the wedge algebra $\M := \A_0(\RR_-)\otimes \A_0(\RR_+)$,
the twisted Borchers triple is given by
\begin{eqnarray*}
\Mtc &:=& \M\otimes\1\vee \Ad e^{i\frac{2\pi k}{N}Q\otimes Q}(\1\otimes \M)\\
&=& \A_0(\RR_-)\otimes\A_0(\RR_+)\otimes\1\vee \Ad e^{i\frac{2\pi k}{N}Q\otimes Q}(\1\otimes \M) \\
&\cong& \left(\A_0(\RR_-)\otimes \1 \vee \Ad e^{i\frac{2\pi k}{N}Q_0\otimes Q_0}(\1\otimes \A_0(\RR_-))\right)\otimes 
\left(\A_0(\RR_+)\otimes \A_0(\RR_+)\right),
\end{eqnarray*}
where in the last line we interchanged the second and the third components for brevity.
We consider the relative commutant of the wedge shifted by $a = (a_+, a_-), a_+ < 0, a_- > 0$.
With the above interchanged notation, it is clear that the $-$-component (the 3rd and 4th tensor components)
of the intersection is simply $\A_0((0,a_-))\otimes\A_0((0,a_-))$.
As for the $+$-component (the 1st and 2nd tensor components), one observes that
it is almost same as the intersection calculated in \cite[Theorem 4.16]{Tanimoto12-2}.
The only change is that the direction of the 2nd component is reversed. This does not affect
the proof, indeed, the inner symmetry commutes with translation and the positivity of energy
is used only through Reeh-Schlieder property hence is not essential.
Thus we have
\[
\Mtc \cap \Ad \Ttc(a)(\Mtc) = \left(\A_0^{\a_0}((a_+,0))\otimes\A_0^{\a_0}((a_+,0))\right)\otimes \left(\A_0((0,a_-))\otimes\A_0((0,a_-))\right),
\]
where $\A_0^{\a_0}$ denotes the fixed point with respect to $\Ad e^{i\frac{2\pi k}{N}Q_0}$.
In particular with $k=1$, this does not satisfies the Reeh-Schlieder property and strict locality fails if
$\ZZ_N$ acts nontrivially.

This counterexample shows that our proof of locality is by no means purely group-theoretic.
Namely, in order to obtain strict locality, it is necessary to assume stronger property than strict locality
itself of the original net (wedge-split property or modular nuclearity as above).
The above massless counterexample appears to be related to the subtlety in massless bootstrap program:
the convergence of form factors is typically worse in massless models and even the fundamental
commutativity theorem relies on the behavior of form factors, which is also worse (c.f. \cite{Smirnov92}).
Hence if one aims at constructing Wightman fields or Haag-Kastler net out of form factors,
the problem of convergence is inevitable.

\section{Realization as deformed fields and scattering theory}\label{scattering}
In Section \ref{inner} we constructed families of Borchers triples
operator-algebraically. Such a construction was also useful for the argument of strict locality
for one case (Section \ref{general}). However, we still have to show that the resulting
nets of von Neumann algebras are really new, or more desirably have nontrivial interaction
and for this purpose the previous presentation is not very convenient. Fortunately, it turns out
that the nets are accompanied by (wedge-local) quantum fields which create one-particle states
and the scattering process can be calculated \cite{BBS01}. We follow the notations of \cite{LS12}.

\subsubsection*{The Zamolodchikov-Fadeev algebra}
Here we consider the construction in Section \ref{inner} applied to the complex free
massive field net. The complex massive free net is given simply by the tensor product
$\Mc := \Mr\otimes\Mr, \Tc := \Trr\otimes \Trr, \Omc := \Omr\otimes\Omr$.
The Hilbert space $\Hc := \Hr\otimes\Hr$ is canonically isomorphic to
the Fock space $\F(\H_1\oplus\H_1)$.
The U(1) symmetry transformation of the complex field is
constructed as follows:
On the ``one-particle space'' $\H_1\oplus\H_1$ we consider the following operator
\[
V_1(\k) \left(\begin{array}{c}\xi\\ \eta\end{array}\right) =
\left(\begin{array}{cc} \cos 2\pi\k & -\sin 2\pi\k \\ \sin 2\pi\k & \cos 2\pi\k \end{array}\right)\left(\begin{array}{c}\xi\\ \eta\end{array}\right),\,\,\, \k \in \RR
\]
The second quantized promotion to $\Hc$ is denoted by $\Vc(\k) := \Gamma(V_1(\k))$.
The operator $\Vc(\k)$ obviously commutes with $\Tc$ and preserves $\Omc$.
Moreover, for field operators one has
\begin{eqnarray*}
\Ad \Vc(\k)(\phi(f)\otimes \1) &=& \cos 2\pi\k(\phi(f)\otimes \1) + \sin 2\pi\k(\1\otimes \phi(f)), \\
\Ad \Vc(\k)(\1\otimes \phi(g)) &=& -\sin 2\pi\k(\phi(g)\otimes \1) + \cos 2\pi\k(\1\otimes \phi(g)),
\end{eqnarray*}
and it holds that
$\Ad \Vc(\k)(e^{i(\phi(f)\otimes\1)}) = e^{i \cos 2\pi\k \phi(f)}\otimes e^{i\sin 2\pi\k \phi(f)}$
and $\Ad \Vc(\k)(e^{i(\1\otimes\phi(g)}) = e^{-i \sin 2\pi\k \phi(g)}\otimes e^{i\cos 2\pi\k \phi(g)}$.
By considering $f$ and $g$ supported in $W_\R$, we conclude that
$\Ad \Vc(\k)(\Mc) = \Mc$. In other words, $\Vc(\k)$ implements an inner symmetry of
the group $\RR/ \ZZ \cong S^1$.

As we have seen in Section \ref{general}, our proof of strict locality
works for an action of $S^1$. Hence in the following we consider only that case.

Since $(\Mr, \Trr, \Omr)$ is wedge-split, so is the tensor product $(\Mc, \Tc, \Omc)$.
By considering the above action of $S^1$ by inner symmetry,
one can construct Borchers triples $(\Mtcs, \Ttc, \Omtc)$ as in Section \ref{inner},
which we know to be strictly local by Section \ref{general}.

We first take a closer look at the action of $\Vc(\k)$. 
The matrix $V_1(\k)$ expressed above can be diagonalized by
$\frac{1}{\sqrt 2}\left(\begin{array}{cc} \1 & i\1 \\ \1 & -i\1\end{array}\right)$ into
$\left(\begin{array}{cc} e^{i2\pi\k} & 0 \\ 0  & e^{-i2\pi\k}\end{array}\right)$.
Correspondingly we define $\H_{1,\pm} := \{\psi\oplus \pm i\psi: \psi \in \H_1\}$.
Then the full Fock space $\Hc$ can be decomposed into $\ZZ$-graded subspaces
$\Hc = \bigoplus_{l\in \ZZ} \Hc^l$ and we may assume that
the generator $\Qc$ of $\Vc(\k)$ acts by $l\1$ on $\Hc^l$
(by definition of the grading), and $\Vc(\k) = e^{i2\pi\k\Qc}$.
Hence on the Hilbert space of our interest $\Htc = \Hc\otimes\Hc =
\bigoplus \Hc^l\otimes\Hc^m$, it is clear that $\Qc\otimes \Qc$ acts by $lm\1$ on $\Hc^l\otimes\Hc^m$.

Now the operator $b^\dagger_+(\psi) := b^\dagger(\psi\oplus(i\psi))$ on $\Hc$
increments the grading and so does $b_-(\psi) := b(\psi\oplus(-i\psi))$.
On the other hand, $b^\dagger_-(\psi) := b^\dagger(\psi\oplus(-i\psi))$ and
$b_+(\psi) := b(\psi\oplus(i\psi))$ decrement the grading.

Now it is easy to see the following twisted commutation relations:
\begin{eqnarray*}
b_\pm^\dagger(\psi_1)\otimes\1\cdot \Ad \Vtcs(\1\otimes b_\pm^\dagger(\psi_2)) -
e^{\pm(\mp i2\pi \k)}\cdot\Ad \Vtcs(\1\otimes b_\pm^\dagger(\psi_2))\cdot b_\pm^\dagger(\psi_1)\otimes\1 &=& 0, \\
b_\pm(\psi_1)\otimes\1\cdot \Ad \Vtcs(\1\otimes b_\pm^\dagger(\psi_2)) -
e^{\mp(\mp i2\pi \k)}\cdot\Ad \Vtcs(\1\otimes b_\pm^\dagger(\psi_2))\cdot b_\pm(\psi_1)\otimes\1 &=& 0, \\
b_\pm^\dagger(\psi_1)\otimes\1\cdot \Ad \Vtcs(\1\otimes b_\pm(\psi_2)) -
e^{\pm(\pm i2\pi \k)}\cdot\Ad \Vtcs(\1\otimes b_\pm(\psi_2))\cdot b_\pm^\dagger(\psi_1)\otimes\1 &=& 0, \\
b_\pm(\psi_1)\otimes\1\cdot \Ad \Vtcs(\1\otimes b_\pm(\psi_2)) -
e^{\mp(\pm i2\pi \k)}\cdot\Ad \Vtcs(\1\otimes b_\pm(\psi_2))\cdot b_\pm(\psi_1)\otimes\1 &=& 0,
\end{eqnarray*}
where the signs $\pm$ etc.\! in the first term correspond to respectively
to $\pm$ etc. in the constant factor in the second term.
The commutation relation between objects with or without $\Ad \Vtcs$ follows trivially from
the usual ones. Namely, we have
\[
[b_\pm^\dagger(\psi_1), b_\pm^\dagger(\psi_2)]  = 0, \;\;\;
[b_\pm^\dagger(\psi_1), b_\pm(\psi_2)]  = \<\psi_2, \psi_1\>\1,\;\;\;
\]
and all other combinations commute
(note that $\<\,\cdot\,,\cdot\,\>$ is linear in the second argument and $b^\dagger(\cdot)$
is linear and $b(\cdot)$ is antilinear).
In other words, these operator-valued distributions satisfy the Zamolodchikov-Fadeev algebra,
with the S-matrix given by the phase factors.
This two-particle scattering matrix (see below) does not depend on the rapidity $\theta = \log p$
(note that $p$ is associated to the lightlike translation, not the spacelike translation as usual.
We will assume that the mass is $1$ for simplicity).

Note that this set of commutation relations can be summarized in the form of matrix.
We take a basis $\{e_{1,+}, e_{1,-}, e_{2,+}, e_{2,-}\}$ on $\CC^2\otimes\CC^2$ and accordingly
$\{e_{1,+}\otimes e_{1,+}, e_{1,+}\otimes e_{1,-}, e_{1,+}\otimes e_{2,+},
e_{1,+}\otimes e_{2,-},\cdots \}$ on $(\CC^2\otimes\CC^2)\otimes (\CC^2\otimes\CC^2)$, where
\[
 (\Hr\oplus\Hr)\oplus(\Hr\oplus\Hr) = (\Hr \otimes \CC^2)\otimes \CC^2
\]
is understood. The signs $\pm$ refer to the structure of the complex free field,
while indices $1,2$ are the first and the second copies of the field.
The two-particle S-matrix $\Stcs(\theta)$ is given
on this basis by (note that this is constant with respect to $\theta$):
\[
\left(\begin{array}{cccc|cccc|cccc|cccc}
1 & & & &              & & & &                                & & & &                           & & & \\
  & & & &            1 & & & &                                & & & &                           & & & \\
  & & & &              & & & &          e^{i2\pi \k}          & & & &                           & & & \\
  & & & &              & & & &                                & & & &    e^{-i2\pi \k}          & & & \\
\hline
 & 1 & & &           &   & & &          &                        & & &   &                       & & \\
 &   & & &           & 1 & & &          &                        & & &   &                       & & \\
 &   & & &           &   & & &          & e^{-i2\pi \k}          & & &   &                       & & \\
 &   & & &           &   & & &          &                        & & &   & e^{i2\pi \k}          & & \\
\hline
 & & e^{-i2\pi \k}          & &    & &                       & &    & &   & &          & &   & \\
 & &                        & &    & & e^{i2\pi \k}          & &    & &   & &          & &   & \\
 & &                        & &    & &                       & &    & & 1 & &          & &   & \\
 & &                        & &    & &                       & &    & &   & &          & & 1 & \\
\hline
 & & & e^{i2\pi \k}          &     & & &                        &   & & &   &          & & & \\
 & & &                       &     & & & e^{-i2\pi \k}          &   & & &   &          & & & \\
 & & &                       &     & & &                        &   & & & 1 &          & & & \\
 & & &                       &     & & &                        &   & & &   &          & & & 1
\end{array}\right)
\]
And it is straightforward that this complies the conditions of \cite[Definition 2.1]{LS12}
if the charge conjugation $\Jc$ is introduced which exchanges $+$ and $-$ fields in each component.
This is said to be diagonal in the sense of \cite[Section 6]{LS12}. The advantage of our methods is
that the strict locality can be seen as an immediate consequence of Section \ref{general}.

This two-particle S-matrix is nontrivial only between different components. One could say
that the interaction occurs only between particles of different species but there is no
self-interaction. This is clear also from the construction: one component remains
unchanged and the other component is just shifted by a unitary equivalence, thus
the twisting exists only between different components.

With the help of the analysis \cite{LS12}, we have the following
since we see below that our von Neumann algebra is generated by those wedge-local fields.
\begin{theorem}
 The triple $(\Mtcs,\Ttc,\Omtc)$ is strictly local and the corresponding Haag-Kastler
net is asymptotically complete and interacting and the S-matrix is factorizing
and its two-particle S-matrix is given as above.
\end{theorem}

\subsubsection*{Comparison of von Neumann algebras}
By definition we have $\Mtcs = \Mc\otimes \1 \vee \Ad \Vtcs(\1\otimes\Mc)$
and $\Mc$ is generated by the exponential of fields $\phi(f)\otimes \1, \1\otimes\phi(g)$, where
$f,g$ are real test functions with $\supp f, \supp g \subset W_\R$ and $\Hc = \Hr\otimes\Hr$ is understood.
First let us consider the $\Mc\otimes\1$ component.
Following \cite{LS12}, we consider pairs of complex valued test functions $f,g$ such that $\overline f = g$.
The complex field which generates the wedge algebra in \cite{LS12} is given in our notation by
\begin{eqnarray*}
\phi_{\mathrm c}(f\oplus g) &=& b_{\mathrm c}^\dagger(f^+\oplus g^+) + b_{\mathrm c}(J_1(f^-\oplus g^-)) \\
&:=& b_+^\dagger(f^+) + b_-^\dagger(g^+) + b_+(\overline{g^-}) + b_-(\overline{f^-}) \\
&=& \phi(f)\otimes\1 + \1\otimes i\phi(f) + \phi(g)\otimes \1 - \1\otimes i\phi(g) \\
&=& \phi(f + \overline f)\otimes \1 + \1\otimes i \phi(f - \overline f) \\
&=& \phi(2\Re f)\otimes \1 - \1\otimes \phi(2\Im f),
\end{eqnarray*}
where $J_1(\xi\oplus \eta) = \overline \eta \oplus \overline \xi$ and $\xi, \eta \in L^2(\RR,d\theta)$
and $f^\pm(\theta), g^\pm(\theta)$ are defined as before. Namely,
the fields in complex and real basis are just the linear combination of each other.
Note that $\phi$ and $b^\dagger$ are linear but $b$ is antilinear. 

It follows also that
$\Ad \Vtcs(\1\otimes\phi_{\mathrm c}(f\oplus g)) = \Ad\Vtcs(\1\otimes(\phi_{\mathrm c}(2\Re f)\otimes\1 - \1\otimes\phi_{\mathrm c}(2\Im f)))$.
From this one easily shows that the wedge algebra generated by the Zamolodchikov-Fadeev
fields in the sense of \cite{LS12} is equal to $\Mtcs = \Mc\otimes\1 \vee \Ad \Vtcs(\1\otimes\Mc)$.

\subsubsection*{Relation to the Federbush model}
One notices that this S-matrix is very similar to the one of the Federbush model
\cite{Ruijsenaars82, CAF01, Schroer98},
although here the fields are bosonic. However, our procedure can be easily adapted to
fermionic nets. Moreover, in the traditional approach there were technical problems:
One can construct local fields only for small coupling constant \cite{Ruijsenaars83}, or
if one takes the bootstrap approach, the convergence of form factors is not clear \cite{BK04}.
Here this problem is completely solved. We can prove the existence of local operators
for any value of $\k$ if we consider the action of $S^1$.
By comparing the S-matrix, this corresponds to an arbitrary value of the coupling constant.

More importantly, our construction is not restricted to the Federbush models.
One can take any wedge-split net with inner symmetry. This contains, for example,
the tensor product of one of Lechner's models \cite{Lechner08}, instead of the real free field.
One can consider $n$ copies of the real free field, which have $\mathrm{O}(n)$ symmetry, then
take any subgroup of $\mathrm{O}(n)$ isomorphic to $S^1$ or $\ZZ_N$. This should correspond to the
Lie-algebraic generalization of the Federbush models, whose form factors were proposed in \cite{CAF01}.
It works also with $n$ copies of one of Lechner's models. 
Furthermore, the constructed net admits again inner symmetry
and is wedge-split, hence one can repeat the construction to obtain further new models
(on a bigger Hilbert space).

\subsubsection*{Some technical remarks on inner symmetry}
The wedge-algebra of our two-dimensional
nets is given by the tensor product twisted by the inner symmetry.
However, this does not mean that there is a subnet which is a copy of a tensor component.
This is clear because any tensor component in the right wedge does not commute with
the left wedge unless it is in the fixed point with respect to the action of $S^1$.
It is also noted that a wedge-split net has no nontrivial DHR sector \cite{Mueger98},
hence any extension of such a net is a tensor product.

One realizes that
the whole net still admits an action of $S^1\times S^1$ by inner symmetry.
The fixed point net may fail to have Haag duality \cite{Mueger98}
and the standard sector theory does not apply.
The whole net is an extension of this fixed point net.


Our models violate the parity symmetry, which is clear from the S-matrix. However, the extended parity which interchanges
the two components is preserved (see \cite[Section 6.3.3]{Schroer98}).
As noted in \cite[Theorem 3.3]{BGL93}, the parity symmetry is essential for the Bisognano-Wichmann
property in two dimensions. Although this is not necessarily related to our models, we present
a simple counterexample. One takes the complex free field, which admits $S^1$-inner symmetry
with Bisognano-Wichmann property. Then one can simply replace the Lorentz boosts by
the composition of Lorentz boosts and the inner symmetry. This still satisfies all the axioms of
net but violates the parity symmetry which must have the appropriate commutation relation with
the boosts.
Accordingly, the inner symmetry and Poincar\'e symmetry do not necessarily commute.
For example, we can take a net with a noncommutative Lie group symmetry and replace the
boosts as above.
The proof in four dimensions \cite[Theorem 10.4]{DL84} does not work in two dimensions
since the Lorentz group is abelian, and hence the Poincar\'e group has finite dimensional
unitary representations.

\section{Borchers triples through Longo-Witten endomorphisms on the $U(1)$-current net}\label{u1}
Here we exhibit another procedure to produce Borchers triples in a more concrete way. We
take the free massive net as the starting point (Section \ref{free}).
The formulae are quite similar to those in \cite{Tanimoto12-2}, but should not be confused.
Strict locality of the models constructed here is not investigated in the present paper.
The author expects that the wedge-local field presentation in Section \ref{real}
would help in order to prove strict locality for the construction here.

\subsection{Reduction to lightray}
Let $(\Ar, \Ur, \Omr)$ be the free massive net. The representation $\Ur$ can be
restricted to the positive lightray, which we denote by $\Ur^+$
and we obtain a one-dimensional Borchers triple $(\Mr, \Ur^+, \Omr)$, namely
a von Neumann algebra $\Mr := \Ar(W_\R)$, a positive energy representation $\Ur^+$ of $\RR$ and
a cyclic separating vector $\Omr$ for $\Ar$ invariant under $\Ur^+$ such that
$\Ad \Ur^+(t)(\Mr) \subset \Mr$ for $t \in \RR_+$.
We denote the translation along the negative lightray by $\Ur^-$.

Let us recall the $\u1$-current net $(\netu1, U_0, \Omega_0)$, which is a conformal net. For its definition,
see our previous discussion \cite[Section 5]{Tanimoto12-2}.
The point is that the Hilbert space is naturally isomorphic to the Fock space $\Hr$
of the massive free net, whose one-particle space is $L^2(\RR,d\theta)$ and
one considers the second quantization operators.
The above one-dimensional triple $(\Mr, \Ur^+, \Omr)$ is actually unitarily equivalent
to the triple $(\netu1(\RR_+), T_0, \Omega_0)$,
where $(\netu1, U_0, \Omega_0)$ is the $\u1$-current net and $T_0$ is the restriction
of $U_0$ to the translation subgroup. This will be explained
in more detail in \cite{BT13}.

In particular, we can exploit the Longo-Witten endomorphisms found in \cite{LW11}.
Recall that, for an inner symmetric function $\f(z)$, namely the boundary value on $\RR$
of a bounded analytic function on $0 < \Im z < \pi$, one considers the operator
$V_\f := \Gamma(\f(P_1))$, where $P_1$ is the generator of the restriction of $\Ur^+$
on the one-particle space $\H_1$, $\f(P_1)$
denotes the operator defined by functional calculus and $\Gamma$ is the second quantization
(note that in general the order of second quantization and functional calculus cannot be
exchanged: $\Gamma(\f(P_1)) \neq \f(\Gamma(P_1))$.
Then $\Ad V_\f$ preserves $\Mr$ and $V_\f$ commutes with $\Ur^+$.
Furthermore, $V_\f$ commutes with $\Ur^-$ since $\Ur^-(a) = \Gamma\left(\exp\left(\frac{it}{P_1}\right)\right)$,
as we see in \cite{BT13}.

\subsection{Construction of Borchers triples}
We work on the tensor product Hilbert space $\Htr := \Hr\otimes\Hr$.
We fix an inner symmetric function $\f$. As above, $P_1$ is the one-particle
lightlike translation. Let us recall our argument \cite[Section 5]{Tanimoto12-2}.

The physical Hilbert space $\Hr$ is included in the unsymmetrized Fock space $\H^\Sigma$.
We consider $m$ commuting operators on $\H_1^{\otimes m}$:
\[
\{\1\otimes\cdots\otimes \underset{j\mbox{-th}}{P_1}\otimes\cdots \otimes \1: 1\le j \le m\}.
\]
For $1 \le j \le m$ and $1 \le k \le n$, let us define operators on $\H^m\otimes\H^n$:
\begin{align*}
 P_{j,k}^{m,n} &:= \underset{j\mbox{-th}}{\left(\1\otimes\cdots\otimes \frac{1}{P_1}\otimes\cdots \otimes \1\right)}
\otimes {\left(\1\otimes\cdots\otimes \underset{k\mbox{-th}}{P_1}\otimes\cdots \otimes \1\right)} \\
R^{m,n}_\f &:= \prod_{j,k}\f(P_{j,k}^{m,n}),
\end{align*}
where $\f(P_{j,k}^{m,n})$ is defined by functional calculus on $\H_1^{\otimes m}\otimes\H_1^{\otimes n}$.
Now, our key operator on the unsymmetrized space $\H^\Sigma \otimes \H^\Sigma$ is
\[
\Rt := \bigoplus_{m,n} R_\f^{m,n} = \bigoplus_{m,n} \prod_{j,k} \f(P^{m,n}_{j,k}),
\]
where for $m=0$ or $n=0$ we set $R_{j,k}^{m,n} = \1$ as a convention.
It is easy to see that $\Rt$ naturally restricts to partially symmetrized subspaces $\Hr\otimes\H^\Sigma$
and $\H^\Sigma\otimes\Hr$ and to the totally symmetrized space $\Hr\otimes\Hr$.

Let $E_1\otimes E_1\otimes\cdots\otimes E_1$ be the joint spectral measure
of $\{\1\otimes\cdots\otimes \underset{k\mbox{-th}}{P_1}\otimes\cdots \otimes \1: 1\le k \le n\}$.
One obtains the following expression:
\begin{eqnarray*}
\f(P_{j,k}^{m,n}) &=& \int \left(\1\otimes\cdots \otimes\underset{j\mbox{-th}}{\f\left(\frac{p_k}{P_1}\right)}\otimes \cdots \1\right)
\otimes \left(\1\otimes\cdots \underset{k\mbox{-th}}{dE_1(p_k)}\otimes\cdots\1\right) \\
&=& \int \left(\1\otimes\cdots \underset{j\mbox{-th}}{dE_1(p_j)}\otimes\cdots\1\right) \otimes
\left(\1\otimes\cdots \otimes\underset{k\mbox{-th}}{\f\left(\frac{P_1}{p_j}\right)}\otimes \cdots \1\right)
\end{eqnarray*}

Similarly to our previous case \cite[Section 5.2]{Tanimoto12-2},
we decompose $\Rt$ with respect only to
the right or left component:
\begin{eqnarray*}
\Rt &=& \bigoplus_n\int \prod_k \G\left(\f\left(\frac{p_k}{P_1}\right)\right) \otimes dE_1(p_1)\otimes\cdots\otimes dE_1(p_n) \\
&=& \bigoplus_m\int dE_1(p_1)\otimes\cdots\otimes dE_1(p_m )\otimes \prod_j \G\left(\f\left(\frac{P_1}{p_j}\right)\right)
\end{eqnarray*}
For the proof, we refer to \cite[Section 5.2]{Tanimoto12-2}.
The expression in the first line naturally restricts to the partially symmetrized space $\Hr\otimes\H^\Sigma$
and the second expression to $\H^\Sigma\otimes \Hr$.

We have a variant of \cite[Lemma 5.2]{Tanimoto12-2}.
\begin{lemma}\label{lm:u1-commutativity}
It holds for $x \in \Ar(W_\R)$ and $x' \in \Ar(W_\R)'$ that
\[
[x\otimes \1, \Ad \Rt (x'\otimes\1)] = 0
\]
on the Hilbert space $\Hr\otimes\Hr$.
Similarly, for $y \in \Ar(W_\R)$ and $y' \in \Ar(W_\R)'$ one has
\[
[\Ad \Rt(\1\otimes y), \1\otimes y'] = 0. 
\]
\end{lemma}
\begin{proof}
The operator $\Rt$ is disintegrated into second quantization operators as we saw above.
First we consider the first of the commutators above.
The operator $\Rt$ restricts naturally to $\Hr\otimes\H^\Sigma$ and $x\otimes\1$
and $x'\otimes\1$ extend naturally to $\Hr\otimes\H^\Sigma$.
It is easy to see that if $\f(z)$ is an inner symmetric function, then so are $\overline{\f(1/\overline z)}$
and hence $\overline{\f(p_k/\overline z)}$ for $p_k \ge 0$.
The first commutation relation is equivalent to 
\[
[\Ad \Rtb(x\otimes \1), x'\otimes\1] = 0.
\]
Let us prove this on $\Hr\otimes\H^\Sigma$. We have
\[
\Ad \Rtb(x\otimes\1)
= \bigoplus_n\int \Ad \left(\prod_k \G\left(\overline\f\left(\frac{p_k}{P_1}\right)\right)\right)(x) \otimes dE_1(p_1)\otimes\cdots\otimes dE_1(p_n),
\]
and this commutes with $x'\otimes\1$. Indeed, since $x\in\Ar(W_\R) = \A_0(\RR_+)$ and $x'\in\Ar(W_\R)' = \A_0(\RR_-)$,
it follows that  $\Ad \G\left(\overline\f\left(\frac{p_k}{P_1}\right)\right)(x) \in \A_0(\RR_+)$
for any $p_k \ge 0$ by the result of
Longo and Witten \cite{LW11} (see also the beginning of this Section),
and by the fact that the spectral support of $E_1$ is positive.
Now the commutation relation just proved naturally restricts to $\Hr\otimes\Hr$ and we obtain
the first relation.

The proof of the second commutation relation goes more simply. We only have to consider the expression
\[
\Ad \Rt(\1\otimes y)
= \bigoplus_m\int dE_1(p_1)\otimes\cdots\otimes dE_1(p_m) \otimes \Ad \left(\prod_j \G\left(\f\left(\frac{P_1}{p_j}\right)\right)\right)(y).
\]
The rest of the argument is parallel as above.
\end{proof}

Let $\Trr$ be the restriction of $\Ur$ to translation (not to be confused with
trace).
Our Borchers triple is given as follows.
\begin{theorem}
The triple
\begin{itemize}
\item $\Mtr := \{x\otimes \1, \Ad \Rt (\1\otimes y): x, y\in\Ar(W_\R)\}''$
\item $\Ttrr = \Trr\otimes \Trr$
\item $\Omtr = \Omr\otimes \Omr$
\end{itemize}
is a Borchers triple.
\end{theorem}
\begin{proof}
The conditions on $\Ttrr$ and $\Omtr$ are readily satisfied since 
they are same as the tensor product net. The operators $\Rt$ and $\Ttrr$ commute since
both are defined by functional calculus of the same spectral measure (recall that
$\Trr(t_+,t_-) = \Ur^+(t_+)\Ur^-(t_-)$), hence
$\Ttrr(t_+,t_-)$ sends $\Mtr$ into itself for $(t_+,t_-) \in W_\R$.
The vector $\Omtr$ is cyclic for $\Mtr$ since
$\Mtr\Omtr \supset \{(x\otimes \1)\cdot \Rt\cdot (\1\otimes y)\cdot\Omtr\}
= \{(x\otimes \1)\cdot (\1\otimes y)\cdot\Omtr\}$
and the latter is total by the Reeh-Schlieder property of the tensor product net.

We see the separating property of $\Omtr$ as follows. Consider a von Neumann algebra
\[
\Mtr^1 := \{\Ad \Rt(x'\otimes \1), \1\otimes y': x', y'\in\Ar(W_\R)'\}''.
\]
One verifies that $\Omtr$ is cyclic for $\Mtr^1$ as above, hence we only have
to show that $\Mtr$ and $\Mtr^1$ commute. This has been done by Lemma \ref{lm:u1-commutativity}.
\end{proof}
One can actually show that $\Mtr^1 = (\Mtr)'$, so this confusing notation is justified.
Indeed, one has only to check the modular group of $\Mtr$ with respect to $\Omtr$ is
the same for the tensor product, which follows from the field representation in the next Section
and the argument of \cite{Lechner11}.

In this presentation we took $\Ar$ as the starting point. It is also possible
to take the models in \cite{Lechner08} or more general models with spectra with more particle
\cite{LS12}. We will discuss this (slight) generalization in \cite{BT13}.
In this paper we do not consider strict locality of these Borchers triples,
although the author expects that a similar proof to \cite{Lechner08} should work.

\subsection{Realization as deformed fields}\label{real}
Let us see that the Borchers triples constructed in Section \ref{u1} admit
a wedge-local field interpretation. Since the twisting operator $\Rt$ is
given on each particle number space, the calculation is straightforward.
We only need the following commutation relations:
\[
b^\dagger(\psi_1)\otimes\1 \cdot \Ad \Rt(\1\otimes b^\dagger(\psi_2)) -
\int dp dp' \psi_1(p)\psi_2(p')\f\left(\frac{p}{p'}\right) \Ad \Rt(\1\otimes b^\dagger(p'))\cdot b^\dagger(p)\otimes \1 = 0,
\]
or with rapidity $\theta = \log p$ (here again $p$ is associated to the lightlike translation as above
and mass is $1$) and in terms of operator-valued distributions one has
\[
 b^\dagger(\theta)\otimes\1\cdot\Ad \Rt(\1\otimes b^\dagger(\theta')) -
\f\left(e^{\theta-\theta'}\right) \Ad \Rt(\1\otimes b^\dagger(\theta'))\cdot b^\dagger(\theta)\otimes \1 = 0.
\]
Note that if $\f$ is an inner symmetric function, then $\f(e^\theta)$ is
a bounded analytic function in the strip $0 < \Im \theta < \pi$ and $\f(e^{i\pi-\theta}) = \f(-e^{-\theta})$
has the same property. In the matrix form, it can be written as
\[
 \Str(\theta) = \left(\begin{array}{cccc}
            1 & & & \\
            & & \f(e^\theta) & \\
            & \f(-e^{-\theta}) & & \\
            & & & 1
           \end{array}\right)
\]
on the basis $\{e_1\otimes e_1, e_1\otimes e_2, e_2\otimes e_1, e_2\otimes e_2\}$.
This is again a diagonal S-matrix. From its simple form, it is expected that
the proof of modular nuclearity is similar to the one in \cite{Lechner08}.

It is clear that also in these models the interaction occurs only between
different components.

\section{Summary and outlook}\label{outlook}
In this paper, we presented a novel procedure to obtain interacting quantum field
models realized as nets of observables. One can view this procedure as first preparing
a pair of models then making them couple. 

The method is thoroughly operator-algebraic and in the most abstract setting of 
Section \ref{general} no field picture (wedge-local or not) is required. On the other
hand, the interaction is fairly weak. When the field description is available, a pair
of particles of different species obtains a phase during the interaction
and no momentum transfer occurs.
A connection with wedge-local field approach and our previously constructed Borchers triples
has been found in \cite{LST12}.
It is worth investigating how to obtain more general
integrable models (e.g. \cite{Lechner08, LS12}) purely operator-algebraically, where
a ``self-interaction'' occurs.

Furthermore, an interesting variant has been obtained in \cite{BT12}. Strict locality of the models
therein is not known, but S-matrix shows a phenomenon which resembles particle production. It is
desired to establish strict locality of these models or to find massive counterparts.

More ambitiously, certain relations are claimed between integrable models and higher-
dimensional gauge theories (e.g.\! \cite{BAA12}). The author wishes to study such connections with
analytical approach.

\subsubsection*{Acknowledgment.}
I would like to thank Marcel Bischoff and Karl-Henning Rehren for helpful discussions.


\def\cprime{$'$}

\end{document}